\documentclass[11pt]{llncs}
\pdfoutput=1

\usepackage{graphicx}
\usepackage{epsfig}
\usepackage{xcolor}
\usepackage{amssymb}
\usepackage{xspace}
\usepackage{multirow}
\usepackage[]{geometry}
\usepackage{subfigure}
\usepackage{wrapfig}

\newtheorem{Lemma}{Lemma}
\newtheorem{Theorem}{Theorem}
\newtheorem{Definition}{Definition}

\def\comment#1{}%
\def\withcomments{%
  \newcounter{mycommentcounter}%
   \def\comment##1{\refstepcounter{mycommentcounter}%
    \ifhmode%
     \unskip%
     {\dimen1=\baselineskip \divide\dimen1 by 2 %
       \raise\dimen1\llap{\tiny
	{-\themycommentcounter-}}}\fi%
     \marginpar[{\renewcommand{\baselinestretch}{0.8}%
       \hspace*{0em}\begin{minipage}{9em}\footnotesize%
[\themycommentcounter]:%
\raggedright ##1\end{minipage}}]{\renewcommand{\baselinestretch}{0.8}%
       \begin{minipage}{9em}\footnotesize%
[\themycommentcounter]: \raggedright%
##1\end{minipage}}}%
  }

\newcommand{\chihhung}[1]{\comment{\textcolor{red}{\textbf{CL:}} #1}}
\newcommand{\andreas}[1]{\comment{\textcolor{green}{\textbf{AG:}} #1}}

\newcommand{\R}{\ensuremath{\mathbb{R}}}
\renewcommand{\P}{\ensuremath{\mathcal{P}}}
\newcommand{\SPM}{\ensuremath{\mathcal{SPM}}}
\newcommand{\FV}{\ensuremath{\mathcal{FV}}}
\newcommand{\kth}{\ensuremath{k^{\mathrm{th}}}\xspace}
\newcommand{\kthorder}{\kth-order\xspace}
\newcommand{\rephrase}[3]{\noindent\textbf{#1~#2.}~\emph{#3}}

\newcommand{\deleted}[1]{}

\begin{document}

\title{Higher Order City Voronoi Diagrams}
\author{Andreas Gemsa$^1$ \and D. T. Lee$^{2,3}$ \and Chih-Hung Liu$^{1, 2}$ \and Dorothea Wagner$^1$}
\institute{Karlsruhe Institute of Technology, Germany \and Academia Sinica, Taiwan \and National Chung Hsing University, Taiwan}

\date{}
\maketitle

\begin{abstract}

We investigate higher-order Voronoi diagrams
in the \emph{city metric}.
This metric is induced by quickest paths in the $L_1$ metric in the presence of an accelerating transportation network of axis-parallel line segments.
For the structural complexity of \kthorder city Voronoi diagrams of $n$ point sites, we show an upper bound of $O(k(n-k)+kc)$ and a lower bound of $\Omega(n+kc)$, where $c$ is the complexity of the transportation network.
This is quite different from the bound $O(k(n-k))$ in the Euclidean metric \cite{Lee-82}.
For the special case where $k=n-1$ the complexity in the Euclidean metric is $O(n)$,
while that in the city metric is $\Theta(nc)$.
Furthermore, we develop an $O(k^2(n+c)\log n)$-time iterative algorithm to compute the $k^{\mathrm{th}}$-order city Voronoi diagram
and an $O(nc\log^2(n+c)\log n)$-time divide-and-conquer algorithm to compute the farthest-site city Voronoi diagram.

\end{abstract}

\section{Introduction}

In many modern cities, e.g., Manhattan, the layout of the road network resembles a grid.
Most roads are either horizontal or vertical, and thus pedestrians can move only either horizontally or vertically.
Large, modern cities also have a public transportation network (e.g., bus and rail systems)
to ensure easy and fast travel between two places.
Traveling in such cities can be modeled well by the \emph{city metric}.
This metric is induced by quickest paths in the $L_1$ metric in the presence of an accelerating transportation network.
We assume that the traveling speed on the transportation network is a given parameter $\nu > 1$.
The speed while traveling off the network is~1.
%Traveling off the network is done with traveling speed~1.
Further, we assume that the transportation network can be accessed at any point.
%Under these circumstances, in the city metric
Then
the distance between two points is the minimum time required to travel between them.

For a given set $S$ of $n$ point sites (i.e., a set of $n$ coordinates) and a transportation network in the plane,
\emph{the $k^{\mathrm{th}}$-order city Voronoi diagram} $V_k(S)$
partitions the plane into \emph{Voronoi regions}
such that all points in a Voronoi region share the same $k$ nearest sites
with respect to the city metric.

The \kthorder city Voronoi diagram can be used to resolve the following situation:
a pedestrian wants to know the $k$ nearest facilities (e.g., $k$ stores, or $k$ hospitals)
such that he can make a well-informed decision as to which facility to go to.
%In this situations \kthorder city Voronoi diagrams can support the decision making process, when
%the geo-coordinates of the facilities are taken as input point sites.
For this kind of scenario,
the \kthorder city Voronoi diagram
provides a way to determine the $k$ nearest facilities, by modeling the facilities as point sites.

%Moreover,
%the \kthorder city Voronoi diagram can be used to cluster a set of facilities
%such as department stores, schools, and hotels into several communities of size $k$
%such that the government can make better city planning according to the distribution of those communities.
%\andreas{I'm not sure the reviewer will get what you mean. it's a bit unclear. what motivations are given in other kthorder VOronoi diagram paper?}
%\andreas{I'm still not really happy with this... I fear without a lot more explanation this is only confusing.}

The nearest-site (first-order) city Voronoi diagram
has already been well-studied \cite{AAP-04,BC-05,BKC-09,GSW-08}.
Its structural complexity (the size)
has been proved to be $O(n+c)$ \cite{AAP-04},
where $c$ is the complexity of the transportation network.
Such a Voronoi diagram can be constructed in $O((n+c)\log(n+c))$ time \cite{BKC-09}.
However, to the best of our knowledge
there is no existing work regarding $k^{\mathrm{th}}$-order
or farthest-site (i.e., $(n-1)^\mathrm{st}$-order) city Voronoi diagrams.

Contrary to \kthorder city Voronoi diagrams,
\kthorder Euclidean Voronoi diagrams
have been studied extensively for over thirty years.
Their structural complexity has been shown to be $O(k(n-k))$ \cite{Lee-82}.
They can be computed by an iterative construction method in $O(k^2n\log n)$ time \cite{Lee-82}
or by a different approach based on geometric duality and arrangements in $O(n^2+k(n-k)\log^2 n)$ time~\cite{CE-87}.
%They can be computed by an iterative construction method in $O(k^2n\log n)$ time and $O(k^2n)$ space \cite{Lee-82}. A different approach to compute such a diagram based on geometric duality and arrangements requires $O(n^2\log n+k(n-k)\log^2 n)$ time and $O(k(n-k))$ space or $O(n^2+k(n-k)\log^2 n)$ time and $O(n^2)$ space \cite{CE-87}.
Additionally, there are several randomized algorithms \cite{ABMS-98,Mulmuley-91}
and on-line algorithms \cite{AS-92,BDT-93}.

One of the most significant differences between the Euclidean metric and the city metric that influences the computation and complexity of Voronoi diagrams is the complexity of a bisector between two points.
In the Euclidean or the $L_1$ metric such a bisector has constant complexity, while in the city metric the complexity may be $\Omega(c)$~\cite{AAP-04} and can even be a closed curve.
Since the properties of a bisector between two points significantly affect the properties of Voronoi diagrams,
a \kthorder city Voronoi diagram can be very different from a Euclidean one.
First, this property makes it non-trivial to apply existing approaches for constructing Euclidean Voronoi diagrams to the city Voronoi diagrams.
Secondly, this property also indicates that the complexity of \kthorder Voronoi diagrams may depend significantly on the complexity of the transportation network.

\begin{table}[tb]
\caption{\small{Comparison between the Euclidean and the city metric. Our results are marked by $^\dag$.}}\label{tb-comparison}
\centering
\hspace*{-1.2 cm}\includegraphics[clip, width=15cm]{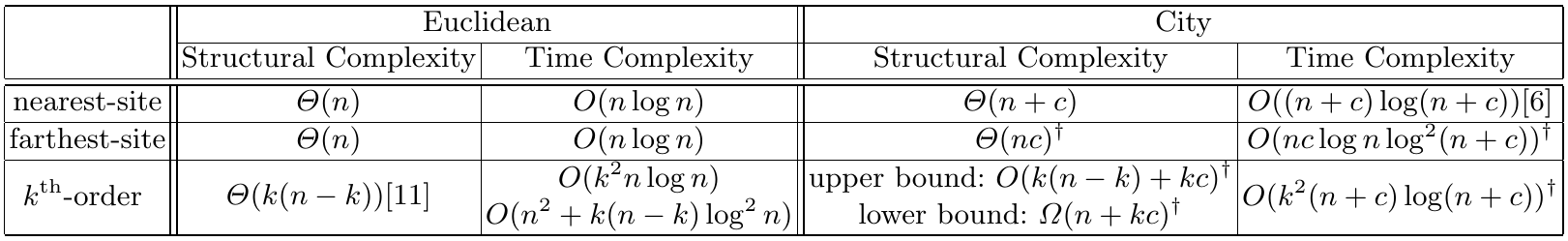}
\end{table}

In this paper, we derive bounds for the structural complexity of the \kthorder Voronoi diagram
and develop algorithms for computing the \kthorder city Voronoi diagram.
The remainder of this paper is organized as follows.
In Section~\ref{sec-prelim},
we introduce two important concepts, wavefront propagation \cite{AAP-04} and shortest path maps \cite{BKC-09},
which are essential for the proofs in the subsequent sections.
In Section~\ref{sec-complexity},
we adopt the wavefront propagation to introduce a novel interpretation of the iterative construction method of Lee \cite{Lee-82},
and use this interpretation to derive an upper bound of $O(k(n-k)+kc)$ for the structural complexity of \kthorder city Voronoi diagrams,
where~$c$ is the complexity of the transportation network.
Then, we construct a worst-case example to obtain a lower bound of $\Omega(n+kc)$.
Finally,
we extend the insights of Section~\ref{sec-complexity} to develop an
iterative algorithm to compute \kthorder city Voronoi diagrams in $O(k^2(n+c)\log(n+c))$ time
(see Section~\ref{sec-algorithms}).
Moreover, we give a divide-and-conquer approach to compute farthest-site city Voronoi diagrams in $O(nc \log n \log^2(n+c))$ time.
We conclude the paper in Section~\ref{sec-conclusion}.

For an overview of our contribution and a comparison between Euclidean and city metric see Table~\ref{tb-comparison}.
%Table~\ref{tb-comparison} shows a comparison of the Euclidean and city metric and summarizes our results.
%Due to the limit of space, we move several proofs to the appendix.
%Due to space constraints, several proofs have been moved to the appendix.

\deleted{
\begin{table}[tb]
\caption{\small{Comparison between the Euclidean and the City metrics. Our results are marked by an asterisk.}}\label{tb-comparison}
\centering

\begin{tabular}{|c||c|c||c|c|}
\hline
 \multirow{2}{*}{}& \multicolumn{2}{c||}{Euclidean} & \multicolumn{2}{c|}{City}\\
 \cline{2-5}
 & Structural Complexity & Time Complexity & Structural Complexity & Time Complexity\\
 \hline
 \hline
 nearest-site & $O(n)$ & $O(n\log n)$ & $O(n+c)$ & $O((n+c)\log (n+c))$ \\
 \hline
 farthest-site & $O(n)$ & $O(n\log n)$ & $\Theta(nc)$* & $O(nc \log n \log^2(n+c))$*\\
 \hline
 \multirow{2}{*}{\kthorder} & \multirow{2}{*}{$O(k(n-k))$} & $O(k^2n\log n)$ & upper bound: $O(k(n-k)+kc)$* & \multirow{2}{*}{$O(k^2(n+c)\log(n+c))$*}\\
  & & $O(n^2+k(n-k)\log^2 n)$  & lower bound:  $\Omega(n+kc)$* & \\
 \hline

\end{tabular}

\end{table}
}

\section{Preliminaries}\label{sec-prelim}
In this section we introduce the notation used throughout this paper
%We first makes denotations and
%assumptions
for \kthorder city Voronoi diagrams.
Then, we introduce two well-established concepts in the context of Voronoi diagrams,
which are important for the proofs in the subsequent sections.

A transportation network is a planar straight-line graph $C = (V_C, E_C)$ with  \emph{isothetic} edges only, i.e., edges that are either horizontal or vertical,
and all transportation edge have identical speed $\nu >1$.
%We use $c$ to denote $|V_C|$, and since the degree of a vertex in $V_C$ is at most four,
We define $c := |V_C|$, and since the degree of a vertex in $V_C$ is at most four,
$|E_C|$ is $\Theta(c)$.
We denote the distance of two points in the $L_1$ metric by $d_1$ and in the city metric by $d_C$.
Similarly, we denote the bisector between two points by $B_1$ and $B_C$ for the $L_1$ and city metric, respectively.
Additionally, for the city metric we define the distance between a point $p \in \R^2$ and a set of points $H \subset \R^2$ to be
$d_C(p, H) = \max_{q \in H}d_C(p,H)$.
This allows us to define the bisector $B_C(H_1,H_2)= \{r\in \R^2 \mid d_C(r,H_1)=d_C(r,H_2)\}$ between two sets of points $H_1$ and  $H_2$.

By $V_k(H,S)$ we denote a \emph{Voronoi region} of $V_k(S)$ associated with a $k$-element subset $H\subset S$.
The common boundary between two adjacent Voronoi regions $V_k(H_1, S)$ and $V_k(H_2, S)$
is called a \emph{Voronoi edge}.
This Voronoi edge is a part of $B_C(H_1,H_2)=B_C(p,q)$
where $H_1\setminus H_2=\{p\}$ and $H_2\setminus H_1=\{q\}$ \cite{Lee-82}.
The common intersection among more than two Voronoi regions
is called a \emph{Voronoi vertex}.
Without loss of generality,
we assume that no point in the plane is equidistant from four sites in $S$ with respect to the city metric,
ensuring that the degree of a Voronoi vertex is
exactly three.

\paragraph{\textbf{Wavefront Propagation.}}
The wavefront propagation
is a well-established model to define Voronoi diagrams \cite{AAP-04}.
In Section~\ref{sec-complexity},
we will use this concept to interpret the formation of $V_k(S)$
and analyze its structural complexity.

For a fixed site $p\in S$, let $W_p(x) = \{q \mid q\in \mathbb{R}^2, d_C(p, q)=x\}$.
This means that for a fixed $x \in \R^+_0$ the wavefront $W_p(x)$ is the circle centered at $p$ with radius $x$.
We call $p$ the \emph{source} of $W_p(x)$.
Note that we can view $W_p(x)$ as the wavefront at time $x$ of the wave that originated in $p$ at time 0.
We refer to such a wavefront as $W_p$ if the value of $x$ is unimportant.

\begin{figure}[t]
\begin{center}
\begin{minipage}[b]{0.5\textwidth}
 \centering
 \includegraphics[width=6.5cm]{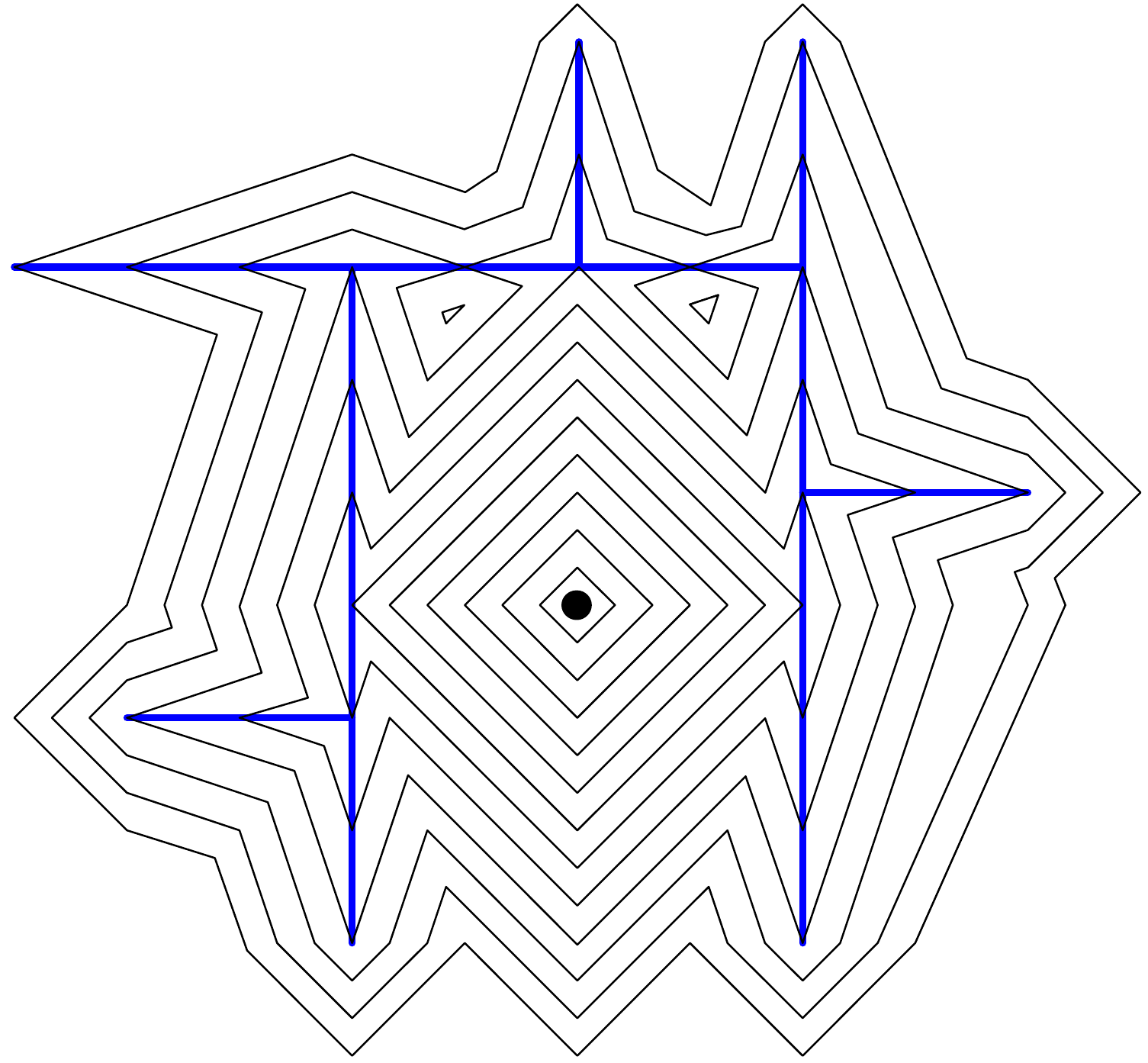}
 \caption{Wavefront Propagation.}
 \label{fig-wavefront}
\end{minipage}
\hfill
\begin{minipage}[b]{0.49\textwidth}
 \centering
 \includegraphics[width=4cm]{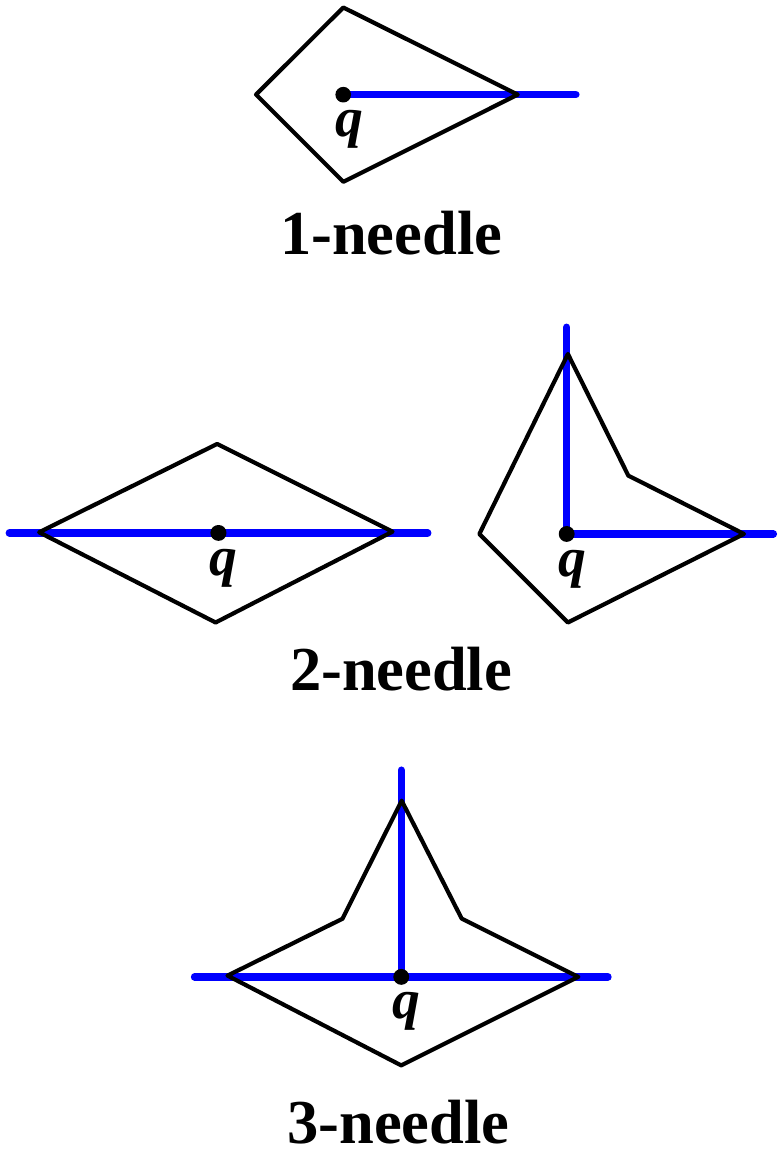}
 \caption{ 1-needle, 2-needle and 3-needle. }
 \label{fig-needles}
\end{minipage}
\end{center}
\end{figure}

\deleted{
\begin{figure}[bt]
\centering
\subfigure[]{\includegraphics[clip, width=8cm]{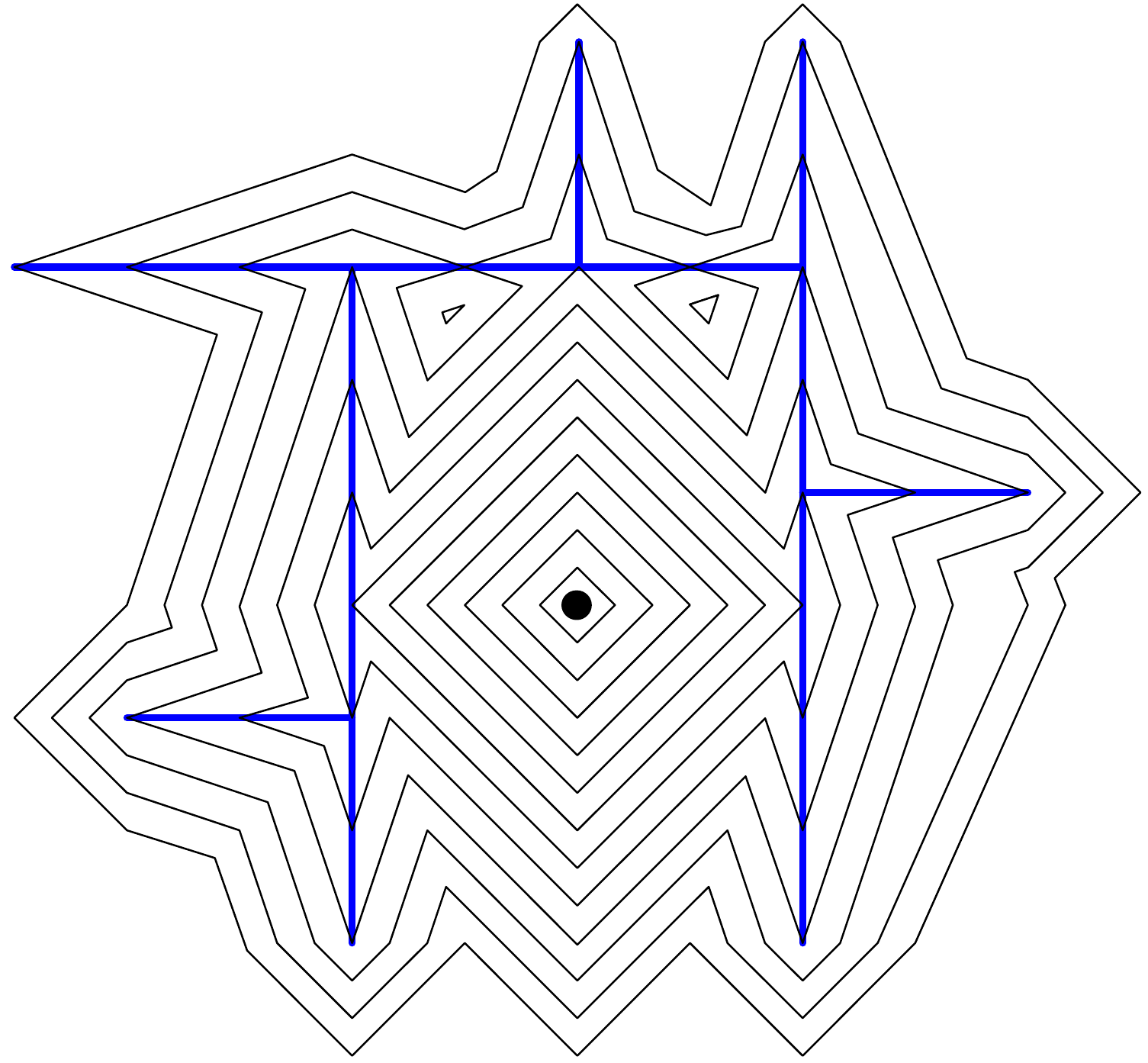}\label{fig-wavefront}}
\subfigure[]{\includegraphics[clip, width=4cm]{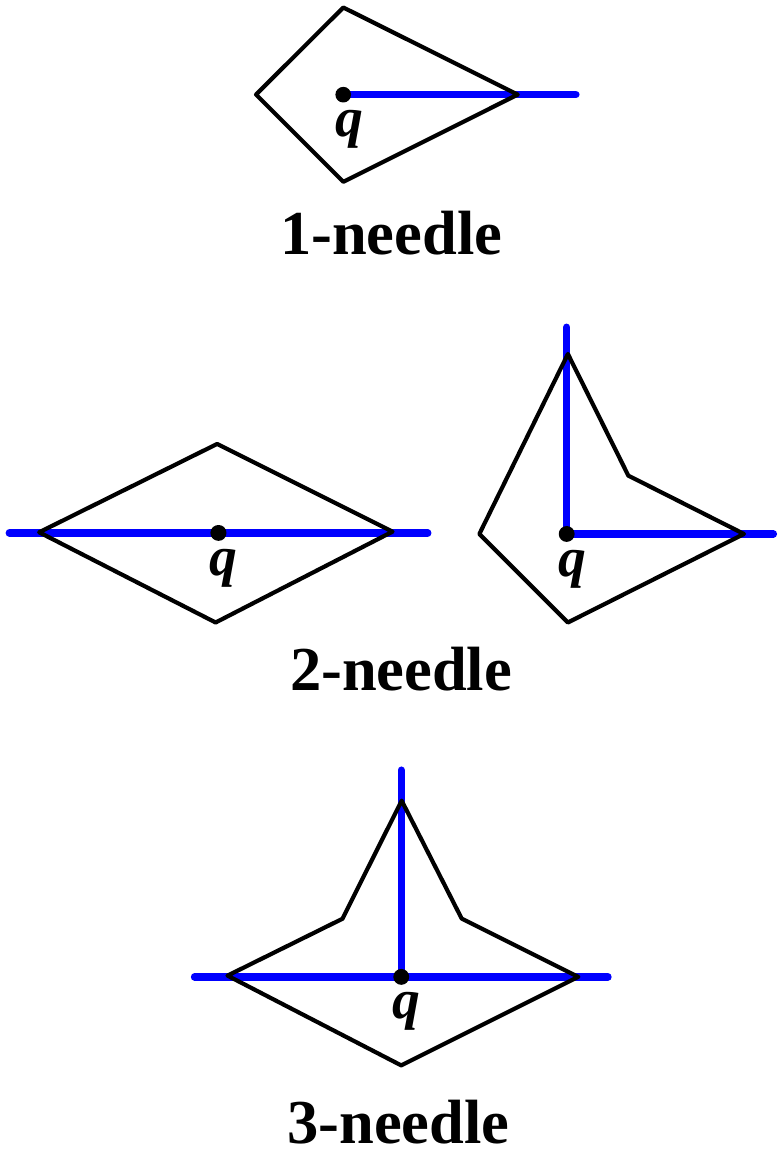}\label{fig-needles}}
\caption{(a) Wavefront Propagation. (b) 1-needle, 2-needle and 3-needle.}
\end{figure}}

Initially, the wavefront $W_p$ is a diamond. When it touches a part of the transportation network for the first time it changes its propagation speed and, hence its shape; see Fig.~\ref{fig-wavefront}.
Certain points on the transportation network play an important role to determine the structural complexity of \kthorder city Voronoi diagrams.
\chihhung{Since each point will be touched a wavefront finally,
``the points on the transportation network that are touched by a wavefront'' is not clear.
I make some changes.}
Thus, we introduce the following definitions.
For a point $v \in \R^2$,
let $P(v)$ denote the \emph{isothetic projection} of $v$ onto the transportation network, i.e.,
we shoot an isothetic half-ray starting at $v$ in each of the four directions and for each half-ray we add its first intersection with an edge of the transportation network to $P(v)$. It is easy to see that there are at most four such intersections.
For a set $X \subset \R^2$, we denote the \emph{isothetic projection of the set $X$} as $\P(X) =\bigcup_{v\in X}P(v)$.
For a site $p \in S$, we call the set $A(p) = P(p) \cup V_C \cup \P(V_C) \cup \{p\}$ \emph{activation points}
(we added $\{p\}$ to the list for ease of argumentation in some of our proofs).

As shown by Aichholzer et al. \cite{AAP-04}, the wavefront
$W_p$ changes its propagation speed only if it hits a vertex in $A(p)$.
Since the shape of $W_p$ can become very complex after it hits multiple activation points,
we make the following simplification for the remainder of this paper:
if a wavefront $W_p$ touches a point $q \in A(p)$ we do not change the propagation speed of~$W_p$.
Instead, we start a new wavefront at $q$, which, in turn,
starts new wavefronts at points in $A(p)$ if it reaches them earlier than any other wavefront.
%Hereafter, the propagation of a new wavefront is called an \emph{activation event}.
Hereafter, the start of the propagation of a new wavefront is called an \emph{activation event}, or we say a wavefront is \emph{activated}.
%New wavefronts activated in this chain reaction are said to originate from $p$.
\chihhung{We did not use the term ``originate'' throughout the paper.
So use it in the proof of Lemma~\ref{lem-mix-region} or remove the sentence.}
The shape of such a new wavefront depends on the position of $q$ on the transportation network.
It can be categorized into one of three different shapes: 1-needle, 2-needle, and 3-needle \cite{AAP-04} (see Fig.~\ref{fig-needles}).
To simplify things, we treat a 2-needle (3-needle) as two (three) 1-needles (see Fig.~\ref{fig-needle-and-bisector}(a)).

\begin{figure}[bt]
\centering
\includegraphics[clip, width=12cm]{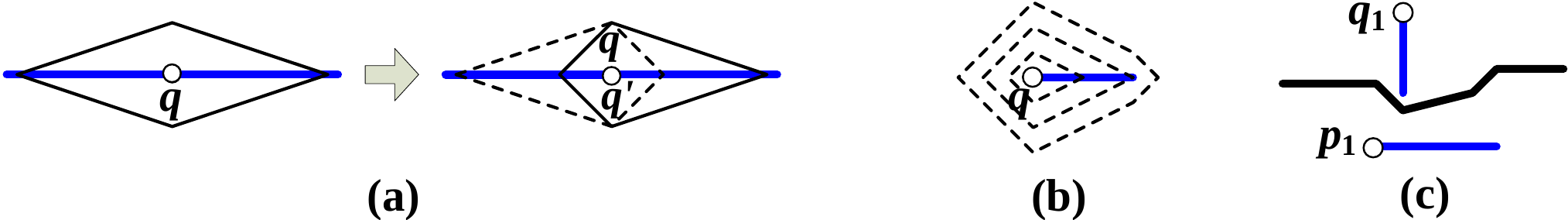}
\caption{(a) A 2-needle is two 1-needles. (b) Wavefront propagation of a 1-needle. (c) $B_1(\eta_p(p_1),\eta_q(q_1))$.}\label{fig-needle-and-bisector}
\end{figure}

When a 1-needle reaches the end of the corresponding network segment,
as shown in Fig.~\ref{fig-needle-and-bisector}(b),
its shape changes (permanently) \cite{AAP-04}.
In order to interpret the propagation of a 1-needle,
Bae et al. \cite{BKC-09} introduced a structure called \emph{needle}.
A needle $\eta_p(q,q')$ is a network segment $\overline{qq'}$ with weight $d_C(p,q)$,
where $p\in S$ and $q,q'\in A(p)$.
Propagating a wavefront from $\eta_p(q,q')$
is equivalent to propagating a 1-needle from $q$ on the network segment $\overline{qq'}$ at time $d_C(p,q)$.
If $q'$ is obvious or unimportant we may refer to $\eta_p(q,q')$ as $\eta_p(q)$.
%They also defined
Bae et al. also defined
the $L_1$ distance $d_1(\eta_p(q,q'),r)$ between a needle $\eta_p(q,q')$ and a point $r$
as $d_C(p,q)$ plus the length of a quickest path from $q$ to $r$ accelerated by $\overline{qq'}$.
Thus, the bisector
$B_1(\eta_p(p_1),\eta_q(q_1))$ between two needles
$\eta_p(p_1)$ and $\eta_q(q_1)$ is well defined (see Fig.~\ref{fig-needle-and-bisector}(c)).

\paragraph{\textbf{Shortest Path Map.}}
We use the wavefront model to define shortest path maps \cite{BKC-09},
and use this concept to explain the formation of \emph{mixed vertices} in Section~\ref{sub-mix},
which are important for deriving the structural
complexity of the \kthorder city Voronoi diagram.

For a site $p\in S$ its shortest path map ${\cal SPM}_p$ is a planar subdivision that can be obtained as follows:
start by propagating a wavefront from the site $p$.
When a point $q\in A(p)$ is touched for the first time by a wavefront,
propagate an additional wavefront from $\eta_p(q)$.
Eventually,
each point $r \in \R^2$ is touched for the first time by a wavefront propagated from a needle $\eta_p(q)$,
where $q\in A(p)$ and $d_1(r, \eta_p(q))=\mbox{min}_{q'\in A(p)}d_1(r, \eta_p(q'))$,
and $q$ is called the \emph{predecessor} of $r$.
This induces $\SPM_p$.
In detail,
$\SPM_p$ partitions the plane into at most $|A(p)|=O(c)$ regions $\SPM_p(q)$
such that all points $r\in\SPM_p(q)$ share the same predecessor $q$ and $q$ is on a quickest path
from $p$ to $r$, i.e., $d_C(p,r)=d_C(p,q)+d_C(q,r)=d_1(r, \eta_p(q))$.
\chihhung{minor rephrasing}
%\andreas{what?}\chihhung{Some changes, but still not be satisfied.}
As proved in \cite{BKC-09}, the common edge between $\SPM_p(q)$ and $\SPM_p(q')$ where $q,q'\in A(p)$
belongs to the bisector $B_1(\eta_p(q),\eta_p(q'))$.
Fig.~\ref{fig-Bisectors} illustrates an example of the function of shortest path maps where
the two Voronoi regions of $V_1(\{p,q\})$ are partitioned by $\SPM_p$ and $\SPM_q$, respectively.
\deleted{Without loss of generality,
in order to guarantee that each vertex of $\SPM_p$ has exactly a degree of 3,
we assume no point in the plane %are equidistant from four needles $\eta_p(q)$, $q\in A(p)$.
has the same distance to four needles $\eta_p(q)$, $q\in A(p)$.
\chihhung{The number of regions in $\SPM_p$ generated by $P(p)$, i.e. $\P(S)$ is trivially at most 4.
Therefore, the number of such regions in all $n$ $\SPM_p$ is $O(n)$.
This is another viewpoint to understand the upper bound.}}

\section{Complexity}\label{sec-complexity}

In this section we
derive an upper and a lower bound of the structural complexity of the \kthorder city Voronoi diagram $V_k(S)$.
In Section~\ref{sub-mix}, we first introduce a special degree-2 vertex on a Voronoi edge called \emph{mixed vertex}
which is similar to the mixed Voronoi vertices of Cheong et al.~\cite{CEGGHLLN-11} for farthest-polygon Voronoi diagrams.
Then we derive an upper bound of the structural complexity of $V_k(S)$ in terms of the number of mixed vertices and Voronoi regions.
In Section~\ref{sub-upper}, we adopt the wavefront concept to introduce a new interpretation for the iterative construction of $V_k(S)$ by Lee \cite{Lee-82}. This yields an upper bound for the structural complexity of $V_k(S)$.
In Section~\ref{sub-lower} we construct a worst-case example to obtain a lower bound for the structural complexity of $V_k(S)$.
\chihhung{We repeat the structural complexity of a \kthorder Voronoi diagram many times above.
Is it possible to give it a short name?}
\andreas{that would be good...}

\subsection{Mixed Vertices}\label{sub-mix}

\deleted
{
\begin{figure}[bt]
\centering
\includegraphics[clip, width=8cm]{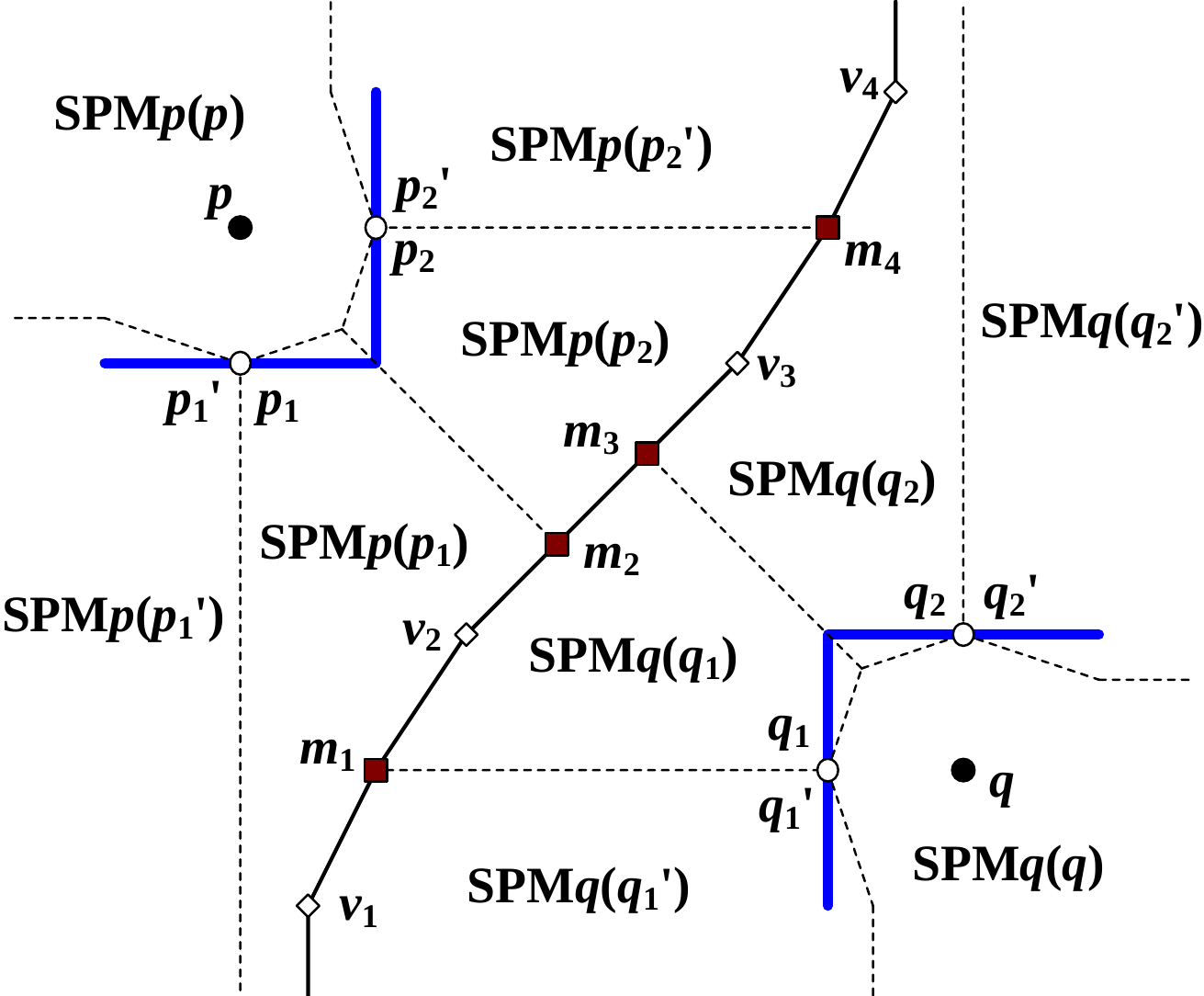}
\caption{ $B_C(p, q)$, where $m_i$ for $1\leq i\leq 4$ is a mixed vertices.}\label{fig-Bisectors}
\end{figure}
}
%\deleted{

\begin{Definition}[Mixed Vertex]\label{def-mix-voronoi-vertex}
For two sites $p,q \in S$ and a Voronoi edge $e$ which is part of $B_C(p, q)$,
a point $r$ on $e$ is \emph{a mixed vertex}
if there are $p_1, p_2\in A(p)$ and $q_1\in A(q)$ such that
$r\in \SPM_p(p_1) \cap \SPM_p(p_2) \cap \SPM_q(q_1)$.
%(Similar for $p_1\in A(p)$ and $q_1, q_2\in A(q)$.)
\end{Definition}

For instance,
Fig.~\ref{fig-Bisectors} shows a first-order city Voronoi diagram $V_1(\{p,q\})$, where the mixed vertices
are marked with a square and denoted by $m_1, \ldots, m_4$.
The vertex $m_2$ is a mixed vertex because it is in $\SPM_p(p_1)\cap \SPM_p(p_2)\cap \SPM_q(q_1)$.
Definition~\ref{def-mix-voronoi-vertex} yields the following.

%The definition given in
%\begin{wrapfigure}{r}{0.5\textwidth}
%\vspace{-0.5cm}\includegraphics[width=.5\textwidth]{figure/Bisector}
%\caption{ $B_C(p, q)$ (solid thin edge), where $m_1, \ldots, m_4$ are mixed vertices.}\label{fig-Bisectors}
%\vspace{-1.35cm}
%\end{wrapfigure}
%}

\begin{figure}
\centering
\includegraphics[width=.5\textwidth]{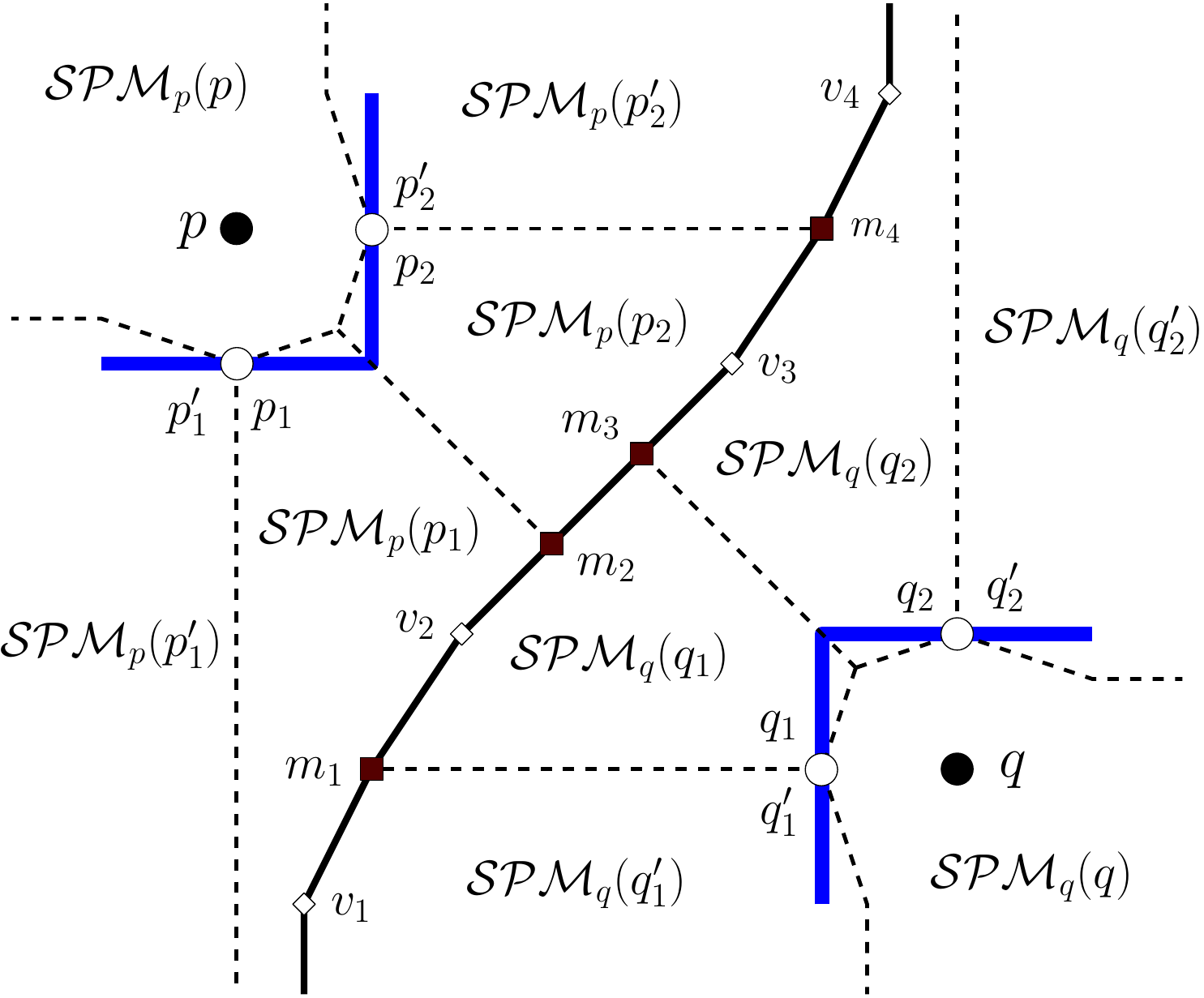}
\caption{ $B_C(p, q)$ (solid thin edge), where $m_1, \ldots, m_4$ are mixed vertices.}\label{fig-Bisectors}
\end{figure}

\newcommand{\lemmVvtext}{%
If a Voronoi edge $e$ contains $m\geq 0$ mixed vertices,
its complexity is $O(m+1)$.
}

\begin{Lemma}\label{lem-mVv}
\lemmVvtext
\end{Lemma}
\begin{proof}
Suppose $e$ is part of a bisector $B_C(p, q)$, where $p, q\in S$.
Consider two consecutive mixed vertices $m_1$ and $m_2$ on $B_C(p, q)$,
where $m_1\in \SPM_p(p_1)\cap\SPM_q(q_1')\cap\SPM_q(q_1)$ and $m_2\in \SPM_p(p_1)\cap \SPM_p(p_2)\cap \SPM_q(q_1)$
(see Fig.~\ref{fig-Bisectors}).
Consider each point $v$ on $B_C(p,q)$ between $m_1$ and $m_2$.
Since $v$ belongs to $\SPM_p(p_1)\cap \SPM_q(q_1)$,
we have $d_C(v,p)=d_C(v,p_1)+d_C(p_1,p)=d_1(v,\eta_p(p_1))$ and $d_C(v,q)=d_C(v,q_1)+d_C(q_1,q)=d_1(v,\eta_q(q_1))$.
Together with $d_C(v,p)=d_C(v,q)$,
$v$ belongs to $B_1(\eta_p(p_1),\eta_q(q_1))$ (recall Fig~\ref{fig-needle-and-bisector}(c)).
As a result, if a Voronoi edge $e$ contains $m$ mixed Voronoi vertices,
$e$ consists of $m+1$ parts, each of which belongs to a bisector between two needles.
Since the complexity of an $L_1$ bisector between two needles is $O(1)$ \cite{BC-05},
the complexity of $e$ is $O(m + 1)$.
Note that the complexity of a bisector between two points in the city metric is $\Theta(c)$,
while the complexity of an $L_1$ bisector between two needles is $O(1)$.
\qed
\end{proof}

\newcommand{\lemmixuppertext}{%
An upper bound for the structural complexity of a \kthorder city Voronoi diagram $V_k(S)$
is $O(M + k(n-k))$, where $M$ is the total number of mixed vertices.
}

\begin{Lemma}\label{lem-mix-upper}
\lemmixuppertext
\end{Lemma}
\begin{proof}
Lee \cite{Lee-82} proved that the number of Voronoi regions in the \kthorder Voronoi diagram is $O(k(n-k))$
in any distance metric satisfying the triangle inequality, and so is the number of Voronoi edges.
By Lemma~\ref{lem-mVv}, if a Voronoi edge $e$ contains $m_e$ mixed Voronoi vertices, its complexity is $O(m_e + 1)$.
Suppose a city Voronoi diagram $V_k(S)$ contains a set $E$ of Voronoi edges,
and each edge $e\in E$ contains $m_e$ mixed Voronoi vertices.
Then, the complexity of all edges, i.e., the structural complexity of $V_k(S)$,
is $\sum_{e\in E} O(m_e + 1) = O(M + |E|)$.
Since $|E|=O(k(n-k))$, it follows that
$O(M + |E|) = O(M + k(n-k))$.
\qed
\end{proof}

For the proof in Section~\ref{sub-upper},
we further categorize the mixed vertices.
Let $m$ be a mixed vertex on the Voronoi edge between $V_k(H_1,S)$ and $V_k(H_2,S)$,
where $H_1\setminus H_2=\{p\}$ and $H_2\setminus H_1=\{q\}$.
We call $m$  an \emph{interior mixed vertex} of $V_k(H_1,S)$
if $m\in \SPM_p(p_1)\cap \SPM_p(p_2)\cap \SPM_q(q_1)$, for some $p_1, p_2 \in A(p)$ and $q_1\in A(q)$;
otherwise, we call $m$ an \emph{exterior mixed vertex} of $V_k(H_1,S)$.
For example, in Fig.\ref{fig-Bisectors}
the vertices $m_2$ and $m_4$ both are interior mixed vertices of $V_1(\{p\}, \{p,q\})$ and exterior mixed vertices of $V_1(\{q\}, \{p,q\})$.

\subsection{Upper Bound}\label{sub-upper}
Throughout this subsection,
we consider a Voronoi region $V_j(H, S)$ of a $j^{\mathrm{th}}$-order Voronoi diagram $V_{j}(S)$,
where $H\subset S$ and $|H|=j$.
Let $V_j(H, S)$ have $h_H$ adjacent Voronoi regions $V_j(H_i, S)$ for $1\leq i\leq h_H$.
Note that the subsets $H_i$ and $H$ differ in exactly one element~\cite{Lee-82}. In the following let
$H_i \setminus H=\{q_i\}$, $Q = \{q_1, \ldots, q_{h_H}\}$, and $\ell_H=|Q|$.

Lee \cite{Lee-82} proved that in any distance metric satisfying the triangle inequality,
$V_j(H, S)\cap V_1(Q) = V_j(H, S)\cap V_{j+1}(S)$,
and thus computing $V_1(Q)$ for all the Voronoi regions $V_j(H, S)$ of $V_j(S)$ yields $V_{j+1}(S)$,
leading to an iterative construction for $V_k(S)$ for any $k< n$.
Fig.~\ref{fig-Euclidean} illustrates this iteration technique for the Euclidean metric:
solid segments form $V_1(H, S)$ and dashed segments form $V_1(Q)$.
Since the gray region is part of $V_1(\{p\}, S)$ and also part of $V_1(\{q_1\}, Q)$,
all points in the gray region share the same two nearest sites $p$ and $q_1$,
implying that the gray region is part of $V_2(\{p, q_1\}, S)$.

\begin{figure}[t]
\begin{center}
\begin{minipage}[b]{0.5\textwidth}
 \centering
 \includegraphics[width=0.7\textwidth]{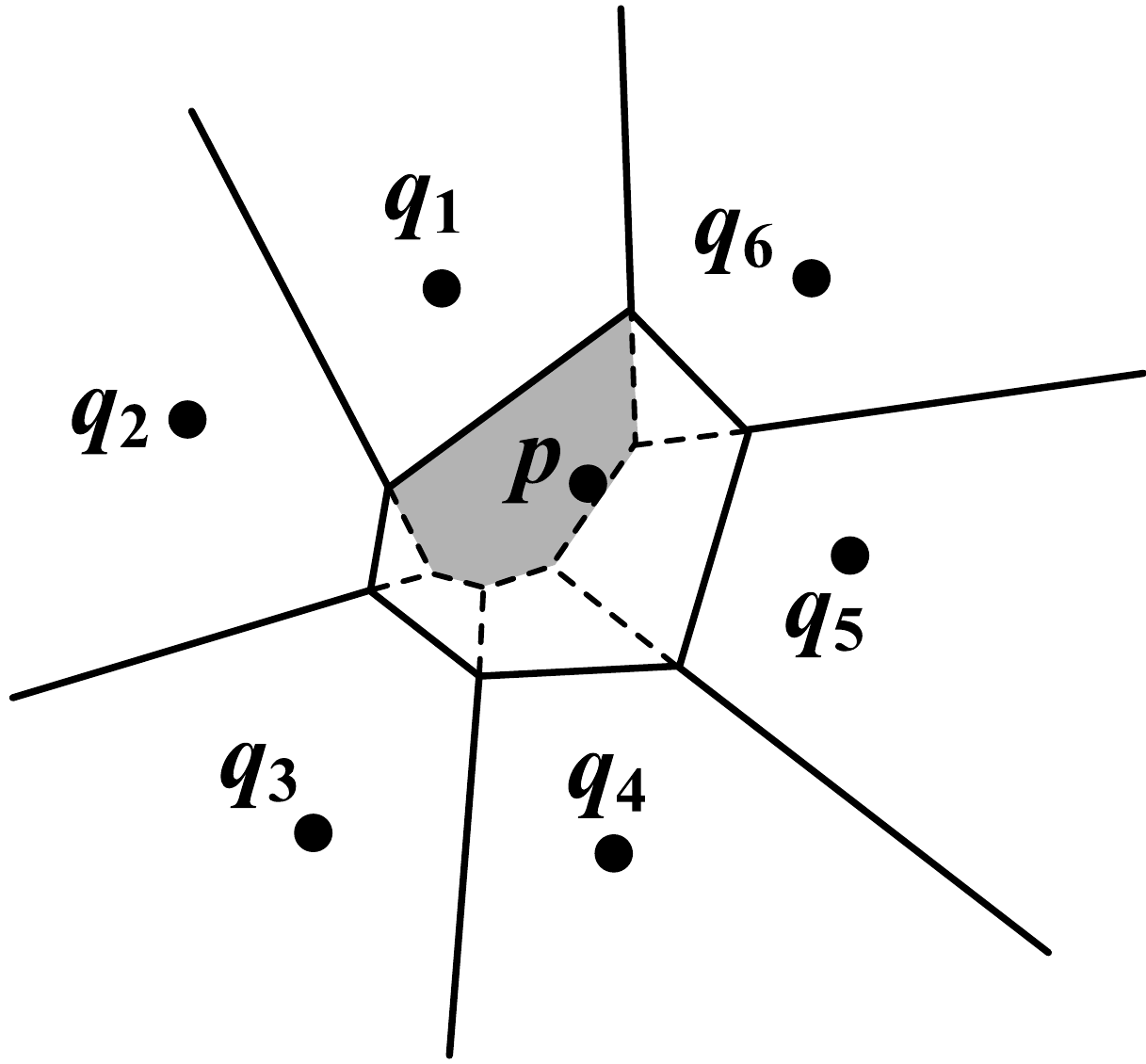}
 \caption{$V_1(H, S)\cap V_1(Q) = V_1(H, S)\cap V_{2}(S)$
 where $H=\{p\}$, $Q=\bigcup_{1\leq i\leq 6} \{q_i\}$, and $S=H\cup Q$.}
 \label{fig-Euclidean}
\end{minipage}
\hfill
\begin{minipage}[b]{0.43\textwidth}
 \centering
 \includegraphics[width=0.7\textwidth]{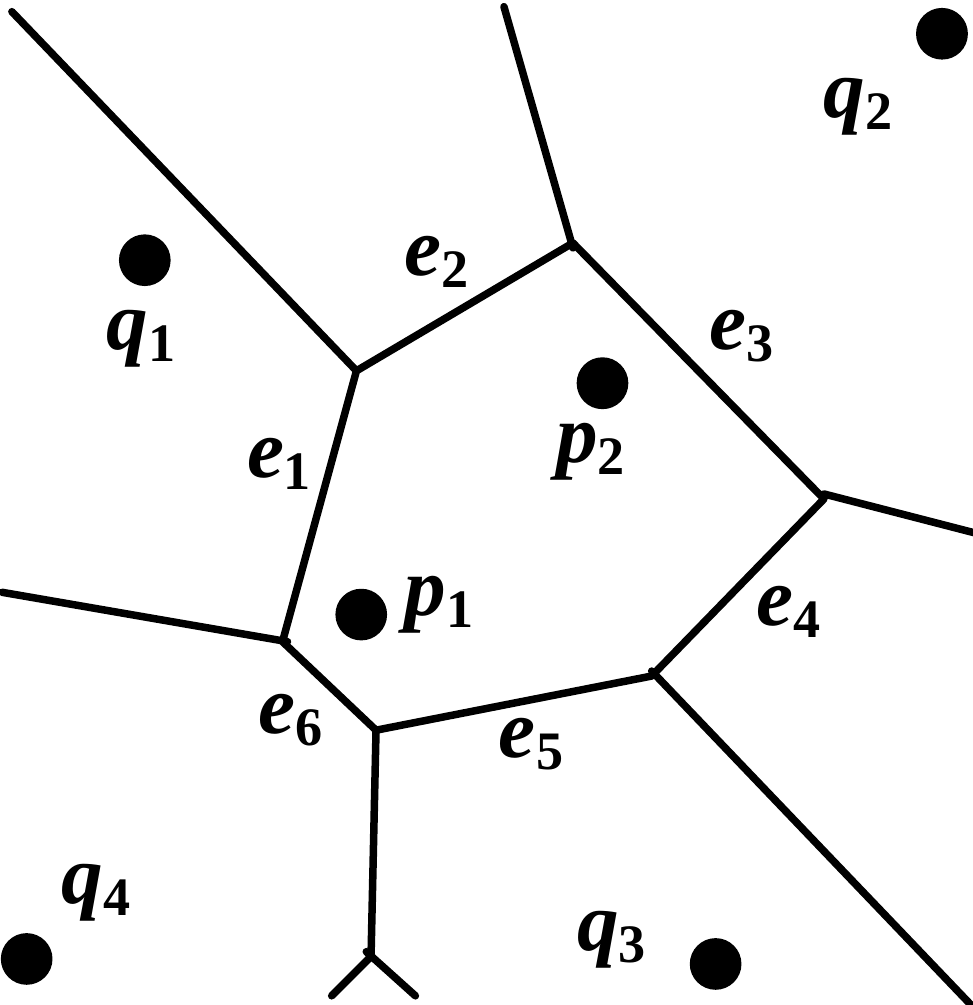}
 \caption{$V_2(H,S)$ where $H=\{p_1, p_2\}$, $Q=\{q_1,q_2,q_3,q_4\}$, and $S=H\cup Q$.  }
 \label{fig-cardinality}
\end{minipage}
\end{center}
\end{figure}

We adopt wavefront propagation to interpret this iterative construction
in a new way, which will lead to the main proof of this section.
Let us imagine that a wavefront is
propagated from each site $q\in Q$ into the Voronoi region $V_j(H,S)$.
If a point $r\in\R^2$ is first touched by the wavefront that propagated from $q$,
$r$ belongs to $V_{j+1}(H\cup\{q\}, S)$.

Note that when $j\geq 2$, $|Q|$ is not necessarily the number of adjacent regions,
i.e., $\ell_H\leq h_H$.
Fig.~\ref{fig-cardinality} illustrates an example for the Euclidean metric:
$V_2(H,S)$ has 6 adjacent Voronoi regions but $|Q|=\ell_H$ is only 4.
This is because for a site $q\in Q$,
$B_2(\{q\}, H)\cap V_j(H,S)$ may consist of more than one Voronoi edge,
where $B_2(\{q\}, H)$ is a Euclidean bisector between $\{q\}$ and $H$ (similar to $B_C(H_1, H_2)$ defined in Section~\ref{sec-prelim}).
For instance, as shown in Fig.~\ref{fig-cardinality},
$e_{q_1}=B_2(\{q_1\},H)\cap V_2(H,S)$ consists of two Voronoi edges $e_1$ and $e_2$

Now we transfer our new interpretation to the city metric.
Let $e_q$ be $B_C(\{q\},H)\cap V_j(H,S)$ for some site $q \in Q$.
If $e_q$ contains $m_q$ \textbf{exterior} mixed vertices with respect to $V_j(H, S)$,
$e_q$ intersects $m_q+1$ regions in ${\cal SPM}_{q}$.
We denote these regions by $\SPM_q(v_z)$ for $1 \leq z \leq m_q + 1$. Note that all $v_z$ must be in $A(q)$. Then,
instead of propagating a single wavefront from $q$ into $V_j(H, S)$ (as in the Euclidean metric), we propagate
$m_q+1$ wavefronts, namely one from each $\eta_{q}(v_z)$ into $V_j(H, S)$.

As a result, if $V_j(H,S)$ contains $m_H$ exterior mixed vertices,
$m_H+\ell_H$ wavefronts will be propagated into $V_j(H,S)$.
During the process, when a point $r\in V_j(H,S)$ is \emph{first} touched by a wavefront
propagated from $\eta_{q}(v), q\in Q \mbox{ and } v\in A(q)$,
we propagate a new wavefront from $\eta_q(r)$, i.e., an activation event occurs,
if (i) $r\in \P(V_C)\cup V_C$
or (ii) $v=q$ and $r\in P(q)$.
These two conditions amount to $r \in A(q)\setminus\{q\}$,
but this classification will help us to derive the number of mixed vertices.
This is due to the fact that during the $k-1$ iterations for computing $V_{j+1}(S)$ from $V_j(S)$ for $1\leq j\leq k-1$,
$\P(V_C)$ contributes $O(kc)$ activation events, but $\P(S)$ only contributes $O(n)$.
%(see the proof of Lemma~\ref{lem-mix-vertices} in the appendix).
%which will be clearly stated in the proof of Lemma~\ref{lem-mix-vertices}.
\chihhung{One problem is that Lemma~\ref{lem-mix-vertices} has be moved to the appendix.}

\newcommand{\lemmixregiontext}{%
If $V_j(H, S)$ contains $m_H$ exterior mixed vertices,
then $V_j(H,S)\cap V_{j+1}(S)$ contains at most $m_H+2c_H+2a_H$ mixed vertices,
where $c_H=|(\P(V_C)\cup V_C)\cap V_j(H, S)|$
and $a_H$ is the number of activation events associated with points in $\P(S)$.
}

\begin{Lemma}\label{lem-mix-region}
\lemmixregiontext
\end{Lemma}

\begin{proof}
According to the above discussion,
we propagate $m_H+\ell_H$ wavefronts into $V_j(H, S)$.
All those wavefronts combined generate at most $c_H$ new wavefronts from points in $\P(V_C)\cap V_j(H_, S)$,
and $a_H$ new wavefronts from points in $\P(S)\cap V_j(H, S)$.
Note that $c_H=|(\P(V_C)\cup V_C)\cap V_j(H, S)|$ (condition~(i)) but $a_H\leq |\P(S)\cap V_j(H, S)|$ (condition~(ii))
\chihhung{This sentence can be removed.}.
\andreas{ok}
\chihhung{maybe move condition~(i) and condition~(ii) to the statement of this lemma.}
Let $W$ be the set of the $m_H+\ell_H+c_H+a_H$ wavefronts.
For each point $r\in V_j(H, S)$, if $r$ is first touched by a wavefront $w \in W$ it is associated with $w$.
This will partition $V_j(H, S)$ into $m_H+\ell_H+c_H+a_H$  regions.
We view those regions as a special Voronoi diagram $V_1(W)$.
Note that $m_H+\ell_H$ of those regions are unbounded.

$V_j(H,S)\cap V_{j+1}(S)$ is a subgraph of $V_1(W)$
since if a point $r\in V_j(H,S)$ is first touched by a wavefront in $W$ propagated from $\eta_q(v)$
\chihhung{or a point $r\in V_j(H,S)$ belongs to $V_1(\eta_q(v),W)$},
$r$ belongs to $V_{j+1}(H\cup\{q\},S)$.
Without loss of generality,
we assume every vertex of $V_1(W)$ has degree~3.
According to Euler's formula it holds that
$N_V=2(N_R-1)-N_U$,
where $N_V$, $N_R$ and $N_U$
are the numbers of vertices, regions, and unbounded regions, respectively.
Since $V_1(W)$ contains $m_H+\ell_H$ unbounded regions and $m_H+\ell_H+c_H+a_H$ bounded regions,
$V_1(W)$ contains $m_H+\ell_H+2c_H+2a_H-2$ vertices.
By \cite{Lee-82},
since $|Q|=\ell_H$,
there are $\ell_H-2$ Voronoi vertices in $V_{j+1}(S)\cap V_j(H,S)$.
Therefore, $V_j(H,S)\cap V_{j+1}(S)$ contains at most $(m_H+\ell_H+2c_H+2a_H-2)-(\ell_H-2)=m_H+2c_H+2a_H$ mixed vertices.
\qed
\end{proof}

\deleted{
By applying Lemma~\ref{lem-mix-region} to each region of $V_j(S)$,
we obtain Lemma~\ref{lem-mix-diagram}. Lemma~\ref{lem-mix-diagram} indicates
a recursive formula: $m_{j+1}=m_j+O(c)+2a_j$, and thus leads to Lemma~\ref{lem-mix-vertices}.
Lemma~\ref{lem-mix-upper} and Lemma~\ref{lem-mix-vertices} gives an upper bound
for the structural complexity of $V_k(S)$ in Theorem~\ref{thm-upper}.
}

Applying Lemma~\ref{lem-mix-region} to each region of $V_j(S)$, yields a recursive formula for the total number of mixed vertices $m_{j+1}$ in $V_{j+1}(S)$: $m_{j+1}=m_j+O(c)+a_j$ (see Lemma~\ref{lem-mix-diagram}). In Lemma~\ref{lem-mix-vertices} we show that this formula can be bounded by $O(n+kc)$ for $k$ iterations of this iterative approach. Finally, in Theorem~\ref{thm-upper} we combine the insights of Lemma~\ref{lem-mix-upper} and Lemma~\ref{lem-mix-vertices} to give an upper bound for the structural complexity of $V_k(S)$. 
%The proofs of Lemmas~\ref{lem-mix-diagram} and \ref{lem-mix-vertices} as well as the proof of Theorem~\ref{thm-upper} can be found in the appendix.

\newcommand{\lemmixdiagramtext}{%
$V_{j+1}(S)$ contains $m_j+O(c)+2a_j$ mixed vertices
where $m_j$ is the number of mixed vertices of $V_{j}(S)$ and
$a_j$ is the number of activation events associated with points in $\P(S)$ during the computation of $V_{j+1}(S)$ from $V_j(S)$.
}

\begin{Lemma}\label{lem-mix-diagram}
\lemmixdiagramtext
\end{Lemma}
\begin{proof}

For a Voronoi region $V_j(H, S)$,
let $m_H$ be the number of its exterior mixed vertices,
let $c_H$ be $|V_j(H, S) \cap (\P(V_C)\cup V_C)|$ and let $a_H$ be number of activation events associated
with vertices in $\P(S)\cap V_j(H, S)$ during the computation.
If $V_j(H, S)$ is empty, $m_H=c_H=a_H=0$.
By Lemma~\ref{lem-mix-region},
$V_j(H, S)\cap V_{j+1}(S)$ contains at most $m_H + 2c_H + 2a_H$ mixed Voronoi vertices.
Therefore, the total number of mixed vertices of $V_{j+1}(S)$ is bounded by:
\[\sum_{H\in S, |H|=j}( m_H+2c_H+2a_H )= m_j+2|\P(V_C)|+2a_j.\]
\qed

\end{proof}

\newcommand{\lemmixverticestext}
{
The number of mixed vertices of $V_k(S)$ is $O(n+kc)$.
}

\begin{Lemma}\label{lem-mix-vertices}
\lemmixverticestext
\end{Lemma}
\begin{proof}
\andreas{to me: proof lacks structure}
\chihhung{I need your help to organize the structure.}
Let $m_j$ be the total number of mixed Voronoi vertices of $V_{j}(S)$
and let $a_j$ be the number of activation events associated with vertices in $\P(S)$ during the computation of $V_{j+1}(S)$ from $V_{j}(S)$, described by our algorithm. Then, by Lemma~\ref{lem-mix-diagram} the following holds:
\[m_k=m_{k-1}+O(c)+a_{k-1}=\cdots=m_1+O(kc)+2\sum_{1\leq j\leq k-1}a_j.\]

Now, we show an upper bound for the complexity of $\sum_{1\leq j\leq k-1}a_j$.
For a vertex $v\in P(q)$ where $q \in S$,
let the $j^{\mathrm{th}}$ iteration be the first time when $v$ is activated by a wavefront propagated from $q$,
i.e. $\eta_q(v)$ will propagate a wavefront,
and let $v$ belong to $V_j(H,S)$.
Due to this and since the points in $H$ are the $j$ nearest sites of~$v$,
$q$ is the $(j+1)^{\mathrm{st}}$ nearest site of $v$.
Therefore, for $j'> j$, if $v\in V_{j'}(H',S)$, $q \in H'$,
implying that $v$ will not be activated by a wavefront propagated from $q$ again after the $j^{\mathrm{th}}$ iteration.
In other words, $v$ causes at most one activation event due to the wavefront propagation of $q$ during the $k-1$ iterations,
and the $\P(S)$ causes $O(n)$ activation events, i.e.,  $\sum_{1\leq j\leq k-1}a_j=O(n)$.
Furthermore, $m_1$ has been proved to be $O(n+c)$\cite{AAP-04,BKC-09,GSW-08}.
Therefore, $m_k=O(n+kc)$. \qed
\end{proof}

\begin{Theorem}\label{thm-upper}
The structural complexity of $V_k(S)$ is $O(k(n-k)+kc)$.
\end{Theorem}

\subsection{Lower Bound}\label{sub-lower}

\deleted{
\begin{figure}
\centering
\includegraphics[clip, width=8cm]{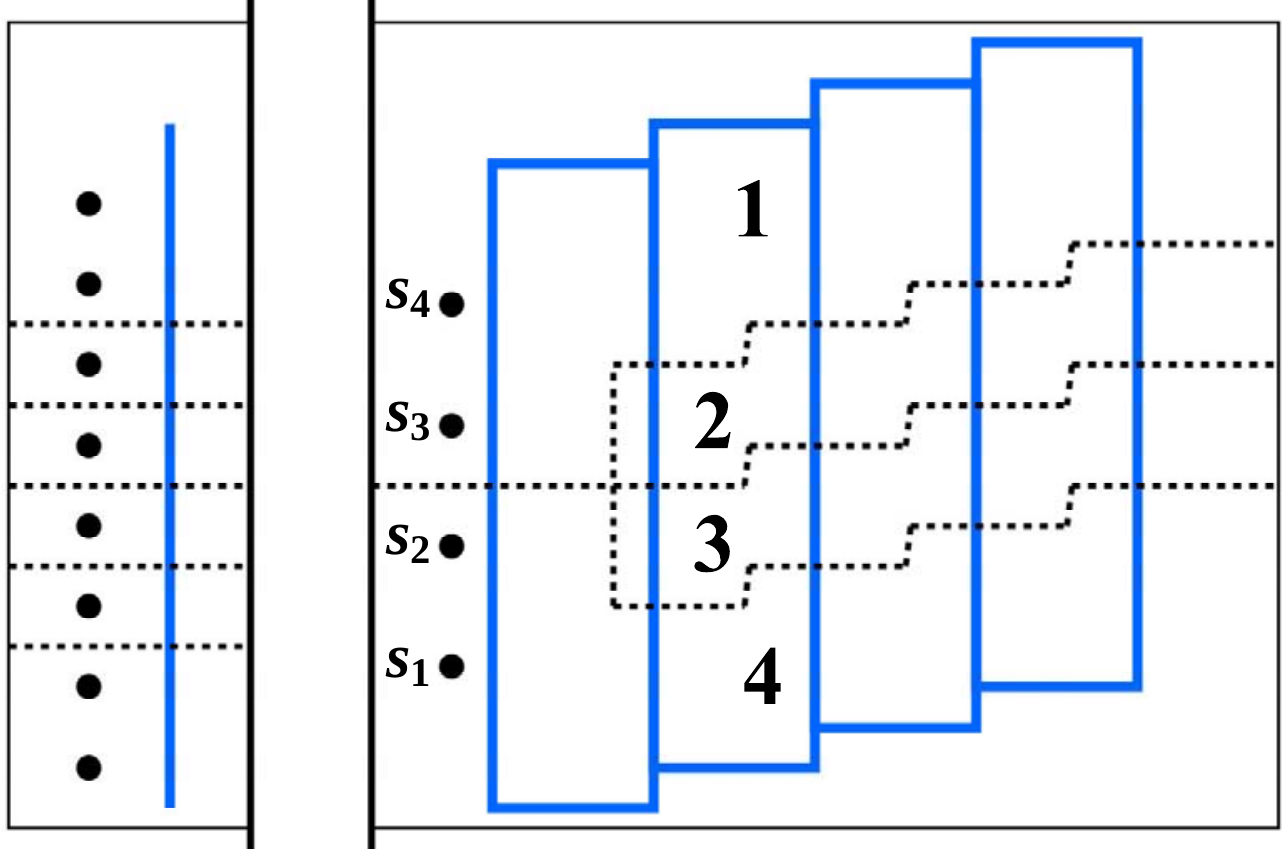}
\caption{A worst-case example
where $k = 3$, $n = 12$, $c = 18$ leads to a lower bound $\Omega(n+kc)$.
The bold solid segments depict the transportation network,
and the dashed segments compose $V_{k}(S)$.
The right part is also the farthest-site city Voronoi diagram of $\{s_1, s_2 , s_3, s_4\}$,
where all points in Region~i share the same farthest site $s_i$.
}\label{fig-worst}
\end{figure}
}

%\andreas{we need a more formal description of how to construct this example. Currently this section is a weird mixture of the example and the description of the construction.}
%\chihhung{I am a little confused. My original thought is to show a small worst-case example and to make the reader believe this example is extendable for arbitrary $n$, $c$, $k$. Alternatively stating the example and the construction would help the reader to understand our ideas. For instance, when we state ``we place one vertical network segment in the left part and build a stairlike transportation network in the right part,'' Fig~\ref{fig-worst} is helpful. This would be because I came from the application area. Maybe you can teach me what to do in the algorithm area.}
We construct a worst-case example (see Fig.~\ref{fig-worst}) to derive a lower bound for the structural complexity of the \kthorder city Voronoi diagram $V_k(S)$.
The example consists of a left part and a right part which are placed with a sufficiently large distance between them.
We place one vertical network segment in the left part
and build a stairlike transportation network in the right part.
Then, we place $k+1$ sites in the right part and the remaining $n-k-1$ sites in the left part.
Since the distance between the left and right part is extremely large,
the $n-k-1$ sites in the left part hardly influence the formation of $V_k(S)$ in the right part.
Therefore, $V_k(S)$ in the right part forms the farthest-site city Voronoi diagram of the $k+1$ sites,
because sharing the same $k$ nearest sites among $k+1$ sites
is equivalent to sharing the same farthest site among the $k+1$ sites.

%\begin{wrapfigure}{r}{.5\textwidth}
\begin{figure}[tb]
\center
\includegraphics[clip, width=.47\textwidth]{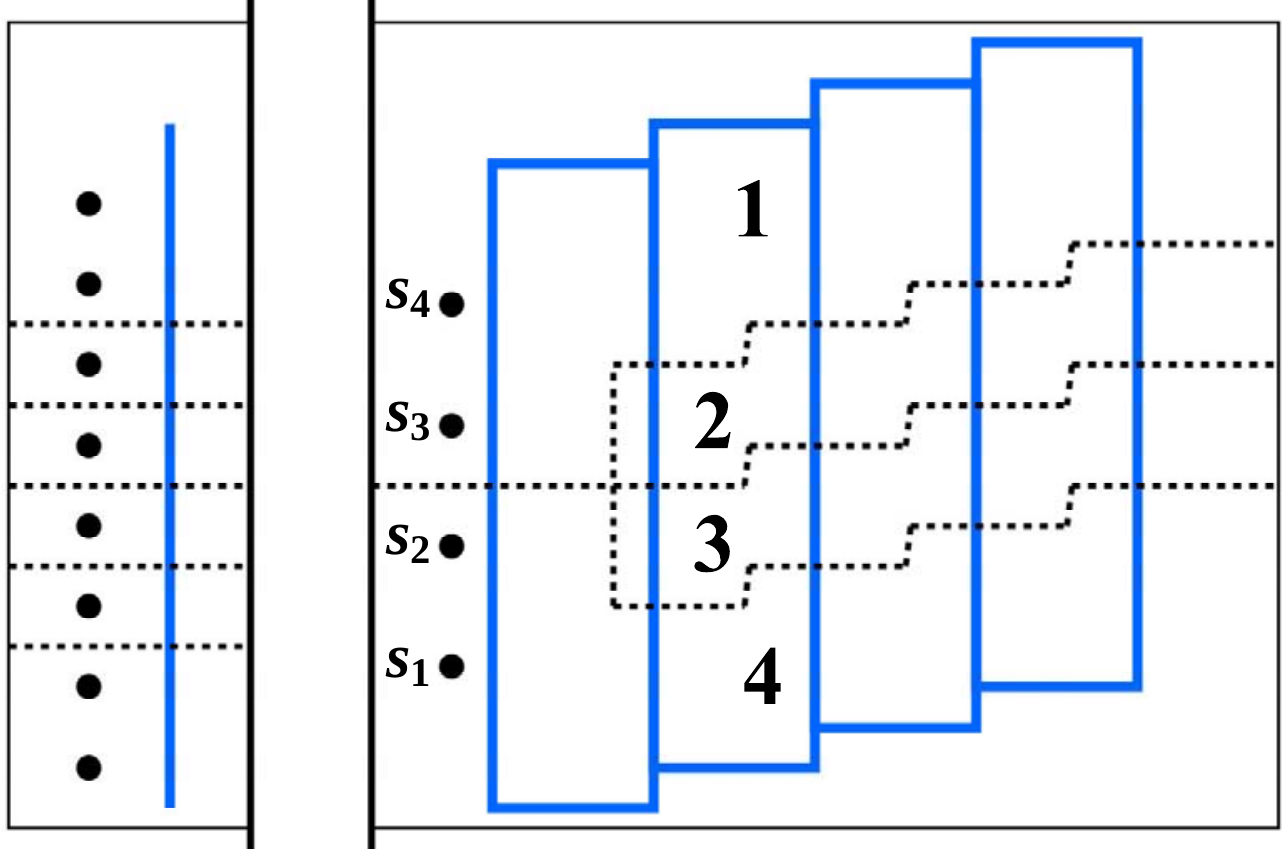}
\caption{This worst-case example
(here with $k = 3$, $n = 12$, $c = 18$) leads to a lower bound of $\Omega(n+kc)$.
The bold solid segments depict the transportation network,
and the dashed segments compose $V_{k}(S)$.
The right part is also the farthest-site city Voronoi diagram of $\{s_1, s_2 , s_3, s_4\}$,
where all points in Region~$i$ share the same farthest site $s_i$.
}\label{fig-worst}
\end{figure}
%\end{wrapfigure}

By construction,
% all points in the right part in a Region $i$ share the same farthest site $s_i$.
as shown in the right part of Fig.~\ref{fig-worst},
all the points in Region~$i$ share the same farthest site $s_i$.
Since we can set the speed $\nu$ to be large enough,
for each point $x$ in Region~2,
the shortest path between $x$ and $s_1$ ($s_2$) moves along the transportation network counterclockwise,
and thus $d_C(s_2,x)> d_C(s_1,x)$.
The common Voronoi edge between Regions $i$ and ($i+1$)
contains at least $(\frac{c-2}{4}-1)\cdot 2+1$ (here: 7) segments
since the transportation network forms $\frac{c-2}{4}$ rectangles
and each rectangle except the first one contains two vertices of the Voronoi edge.
Therefore, in the right part,
$V_k(S)$ contains at least $(k-1)\frac{c-6}{2}=\Omega(kc)$ segments.
Together with the $\Omega(n-k)$ in the left part, we obtain the following lower bound.
%\andreas{minor hint for extending this to k $> 3$}
%\chihhung{I remove the hint because what I mentioned before is wrong.
%We do not need to insert sites between the 4 sites for $k>3$.
%We just construct a new example and put the $k+1$ vertices in the same way.}

\newcommand{\thmlowerdiagramtext}{The structural complexity of $V_k(S)$ is $\Omega(n+kc)$.
}

\begin{Theorem}\label{thm-lower-diagram}
\thmlowerdiagramtext
\end{Theorem}
\begin{proof}
We need to distinguish two cases:

i) $n-k-1\geq k+1 \Rightarrow 2k\leq n-2$:
This implies that the left part contains $n-k-2$ segments, and thus, $V_k(S)$ contains $n-k-2+\Omega(kc)=\Omega(n+kc)$ segments.

ii) $2k> n-2$:
This implies that the left part is empty, and thus,
$V_k(S)$ contains $\Omega(kc)=\Omega(k(c+2))=\Omega(n+kc)$ segments.

This concludes the proof. \qed
\end{proof}

%We further explain the shape of a bisector $B(s_i, s_{i+1})$ between $s_i$ and $s_{i+1}$, for $1\leq i\leq 3$, in Fig.~\ref{fig-worst}.
%The stairlike transportation network in right part forms four rectangles.
%Except the leftmost one, each bisector $B(s_i, s_{i+1})$ in each rectangle forms a Z-shape path.
%Consider the second rectangle $R_2$ to the left, and the Z-shape path of $B(s_1, s_2)$ inside $R_2$.
%Let $e_l$ and $e_r$ be the left edge and the right edge of $R_2$, respectively,
%and let the Z-shape path consist of $e_1$, $e_2$, and $e_3$ from left to right.
%First, for each point $v\in e_1$, both the quickest paths from $s_1$ to $v$ and from $s_2$ to $v$ make use of the transportation edge $e_l$
%(counterclockwise and clockwise, respectively).
%Then, for each $v\in e_2$,  the quickest path from $s_2$ to $v$ still makes use of $e_l$,
%while the quickest path from $s_1$ to $v$ makes use of $e_r$.
%Finally, for each $v\in e_3$, both the quickest paths from $s_1$ to $v$ and from $s_2$ to $v$ make use of $e_r$.

\section{Algorithms}
\label{sec-algorithms}

In this section we present an iterative algorithm to compute \kthorder city Voronoi diagrams in
$O(k^2(n+c)\log (n+c))$ time. Its main idea has already been introduced in the complexity considerations in Section~\ref{sub-upper}.
For the special case of the
farthest-site Voronoi diagram, i.e., the $(n-1)^{\mathrm{st}}$-order Voronoi diagram,
this algorithm takes $O(n^2(n+c)\log(n+c))$ time. However, for the farthest-site city Voronoi diagram we present a divide-and-conquer algorithm which requires only $O(nc\log^2(n+c)\log n)$ time.

\subsection{Iterative Algorithm for $k^{\mathrm{th}}$-Order City Voronoi Diagrams}

We describe an algorithm to compute \kthorder city Voronoi diagrams $V_k(S)$ based on
the ideas in Section~\ref{sub-upper} and
Bae et al.'s \cite{BKC-09} $O((n+c)\log(n+c))$-time
algorithm for the first-order city Voronoi diagram $V_1(S)$.
Bae et al.'s approach views each point site in $S$ as a needle with zero-weight and zero-length,
and simulates the wavefront propagation from those needles to compute $V_1(S)$.
Since their approach can handle general needles,
we adopt it to simulate the wavefront propagation of Section~\ref{sub-upper} to compute $V_{j+1}(S)$ from $V_j(S)$.

\paragraph{Algorithm.}
We give the description of our algorithm for a single Voronoi region $V_{j}(H, S)$.
All four steps have to be repeated for each Voronoi region of $V_j(S)$.

Let $V_{j}(H, S)$ have $h$ adjacent regions $V_{j}(H_i, S)$ with $H_i \setminus H = \{q_i\}$ for $1\leq i\leq h$
and let $Q = \bigcup_{1\leq i\leq h} q_i$.
Our algorithm computes $V_{j}(H, S)\cap V_{j+1}(S)$ as follows:
\begin{enumerate}
  \item Compute a new set $N$ of sites (needles):
        For $1\leq i\leq h$, if the Voronoi edge between $V_j(H_i, S)$ and $V_{j}(H, S)$ intersect $m_i$ regions $\SPM_{q_i}(v_z)$ in $\SPM_{q_i}$, $1 \leq z \leq m_i$, insert every $\eta_{q_i}(v_z)$ into $N$.
  \item Construct a new transportation network $C_H$ from $C$:
        For each point $v \in (\P(V_C)\cup \P(Q) \cup V_C)\cap V_{j}(H, S)$,
        if $v$ is located on an edge $e$ of $C$, insert $e$ into $C_H$.
  \item Perform Bae et al.'s wavefront-based approach
         to compute $V_1(N)$ under the new transportation network $C_H$.
         The approach can intrinsically handle needles as weighted sites.
  \item Determine $V_{j}(H, S)\cap V_{j+1}(S)$ from $V_1(N)$:
		Consider each edge $e$ in $V_{j}(H, S) \cap V_1(N)$.
        Let~$e$ be an edge between $V_1(\eta_p(v_p), N)$ and $V_1(\eta_q(v_q), N)$
        where $p, q\in S$, $v_p\in A(p)$ and $v_q\in A(q)$.
        If $p\neq q$, then $e\cap V_{j}(H, S)$ is part of $V_{j}(H, S)\cap V_{j+1}(S)$.
\end{enumerate}

Note that Step 2 is used only to reduce the runtime of the algorithm.
Lemma~\ref{lem-compute-region} shows the correctness and the run time of this algorithm for a single Voronoi region.

\newcommand{\lemcomputeregion}{%
$V_{j}(H, S)\cap V_{j+1}(S)$ can be computed in $O((h+m+c_H)\log (n+c))$ time,
where~$h$ is the number of Voronoi edges, $m$ is the number of mixed vertices,
and $c_H=|(\P(V_C)\cup V_C)\cap V_{j}(H, S)|$.
}

\begin{Lemma}\label{lem-compute-region}
\lemcomputeregion
\end{Lemma}
\begin{proof}
We begin by proving correctness.
Since $V_{j}(H, S)\cap V_{j+1}(S)$ is exactly $V_{j}(H, S)\cap V_1(Q)$ \cite{Lee-82},
it is sufficient to prove that the algorithm correctly computes $V_{j}(H, S)\cap V_1(Q)$.
If the algorithm fails to compute $V_1(q, Q)\cap V_{j}(H, S), q\in Q$,
it must fail to propagate a wavefront from a needle $\eta_q(v)$, where $v$ belongs to $A(q)$ and
$\SPM_q(v)\cap V_1(q, Q)\cap V_{j}(H, S)$ is nonempty.
We prove that this cannot occur by contradiction. Assume that the algorithm does not propagate a wavefront from an $\eta_q(v)$ for some $v$ in $A(q)$ and $\SPM_q(v)\cap V_1(q, Q)\cap V_{j}(H, S)$ is nonempty.
However, either $v\notin V_{j}(H, S)$ and $\SPM_q(v)\cap B_C(q, p)\cap V_{j}(H, S)$ must be nonempty, then,
Step~1 will include $\eta_q(v)$ in $N$.
Or $v \in V_{j}(H, S)$, then
Step~2 will include the corresponding network segment in $C_H$,
and thus $\eta_q(v)$ will be activated to propagate a wavefront.
Both possibilities contradict the initial assumption.
Therefore, the algorithm correctly computes $V_{j}(H, S)\cap V_{j+1}(S)$.

We proceed by giving time complexity considerations.
It is clear that $|N|$ is $O(m+h)$.
The run time of step~1 is linear in the complexity of the boundary of $V_{j}(H, S)$
and thus is $O(m+h)$.
Since by definition $|(\P(V_C)\cup V_C)\cap V_{j}(H, S)| = O(c_H)$ and $|\P(Q)|=O(h)$,
both $|V_{C_H}|$ and $|E_{C_H}|$ are in $O(c_H+h)$.
Since $|N|=O(m+h)$ and  $E_{C_H}=O(c_H+h)$,
Step~3 takes $O((h+m+c_H)\log (h+m+c_H))$ time \cite{BKC-09}.
Step~4 takes the time linear in the complexity of $V_{1}(N)\cap V_{j+1}(S)$.
The activation events associated with vertices in $\P(S)$ are only associated to vertices in $\P(Q)$, we know that $|\P(Q)| = O(h)$.
Therefore, since there are $O(m+h)$, $O(c_H)$, and $O(h)$ wavefronts
due to $N$, $(\P(V_C)\cup V_C)\cap V_{j}(H, S)$, and $\P(Q)$, respectively,
the complexity of $V_{1}(N)$ is $O(m+h)+O(c_H)+O(h)=O(h+m+c_H)$.
Since $m=O(n+kc)=O(nc)$ it holds that $O(\log(h+m+c_H))=O(\log(nc))=O(\log(n+c)^2)=O(\log(n+c))$.
We conclude that the total running time is  $O((h+m+c_H)\log (n+c))$.
\qed
\end{proof}

Applying Lemma~\ref{lem-compute-region} to each region of $V_j(S)$ combined with Theorem~\ref{thm-upper}
leads to Lemma~\ref{lem-compute-diagram}.
The summation of $O((j(n-j)+jc)\log(n+c))$ in Lemma~\ref{lem-compute-diagram} for $1\leq j\leq k-1$
gives Theorem~\ref{thm-kth-time}.

\newcommand{\lemcomputediagramtext}
{
$V_{j+1}(S)$ can be computed from $V_j(S)$ in $O((j(n-j)+jc)\log(n+c))$ time.
}

\begin{Lemma}\label{lem-compute-diagram}
\lemcomputediagramtext
\end{Lemma}
\begin{proof}
For a Voronoi region $V_j(H, S)$,
let $h_H$ be the number of Voronoi edges,
$m'_H$ be the number of mixed vertices
\chihhung{the prime is to distinguish from the ``exterior'' mixed vertices.},
and $c_H$ be $|V_j(H, S) \cap (\P(V_C)\cup V_C)|$.
By Lemma~\ref{lem-compute-region},
the time complexity of computing $V_{j+1}(S)$ from $V_j(S)$ is
\[\sum_{H\in S, |H|=j}( (h_H+m'_H+c_H)\log (n+c)).\]
By Theorem~\ref{thm-upper},
$\sum_{H\in S, |H|=j}h_H + m'_H =O(j(n-j)+jc)$
\chihhung{Although each mixed vertices will be counted twice during the summation,
it still holds.}.
It is also clear that $\sum_{H\in S, |H|=j}c_H=O(c)$.
Therefore,
the total time complexity is $O((j(n-j)+jc)\log(n+c))$.
The correctness follows from the correctness proof of Lemma~\ref{lem-compute-region}.\qed
\end{proof}

\newcommand{\thmkthtimetext}{
$V_k(S)$ can be computed in $O(k^2(n+c)\log(n+c))$ time.
}

\begin{Theorem}\label{thm-kth-time}
\thmkthtimetext
\end{Theorem}
\begin{proof}
By Lemma~\ref{lem-compute-diagram},
the total time complexity is $\sum_{i=1}^{k-1} O((i(n-i)+ic)\log(n+c))=O(k^2(n+c)\log(n+c))$.
\qed
\end{proof}

\subsection{Divide-and-Conquer Algorithm for Farthest-Site City Voronoi Diagram}

%In this section we describe a divide-and-conquer approach to compute the farthest-site city Voronoi diagram based on work by Cheong et al.'s \cite{CEGGHLLN-11}.
In this section we describe a divide-and-conquer approach to compute the farthest-site city Voronoi diagram $\FV(S)$.
%\chihhung{We will give a full description for \cite{CEGGHLLN-11} two paragraphs later,
%so I think it is not necessary to mentioned it in advance.
%Besides, we already mentions $V_{n-1}(S)=\FV(S)$ in the beginning of Section~\ref{sec-algorithms}.}
%\andreas{ok}
%Note that the farthest site Voronoi diagram is the $(n-1)^\mathrm{st}$-order Voronoi diagram and hence,
Since there are $n$ Voronoi regions in $\FV(S)$ and each of them is associated with a site $p \in S$,
we denote such a region by $\FV(p, S)$.

The idea behind this algorithm is as follows: To compute ${\cal FV}(S)$,
divide $S$ into two equally-sized sets $S_1$ and $S_2=S\setminus S_1$,
%according to $x$-$y$-coordinates (for each $p\in S_1$, $q\in S_2$, $(x_p,y_p)\leq (x_q, y_q)$),
compute ${\cal FV}(S_1)$ and ${\cal FV}(S_2)$,
and then merge the two diagrams into ${\cal FV}(S)$.
Now, suppose we have already computed ${\cal FV}(S_1)$ and ${\cal FV}(S_2)$.
Then, the edges of a Voronoi region ${\cal FV}(p, S)$ in $\FV(S)$ stem from three sources: i) contributed by $\FV(S_1)$, ii) contributed by $\FV(S_2)$, and iii) contributed by two points, one in $S_1$ and the other in $S_2$, that have the same distance to two farthest site.
In fact, the union of all of the third kind of edges is $B_C(S_1, S_2)$.
Each connected component of $B_C(S_1, S_2)$
is called a \emph{merge curve}. A merge curve can be either a closed or open simple curve.
%\andreas{simple = non-intersecting?}
%\chihhung{Yes, no self-intersecting.}

If all the merge curves are computed,
merging ${\cal FV}(S_1)$ and ${\cal FV}(S_2)$
takes time linear in the complexity of $B_C(S_1, S_2)$.
%There are two steps to compute merge curves:
\deleted{
\begin{enumerate}
  \item Find at one point on each merge curve.
  \item Trace out the merge curve from the discovered point.
\end{enumerate}
}
%i) Find at one point on each merge curve, and ii) trace out the merge curve from the discovered point.
To compute the merge curves, we first need to find a point on each merge curve, and then trace out the merge curves from these discovered points.

In order to compute a merge curve,
we modify Cheong et al.'s divide-and-conquer algorithm \cite{CEGGHLLN-11}
for farthest-polygon Voronoi diagrams in the Euclidean metric to satisfy our requirements.
Given a set $\P$ of disjoint polygons, $\P = \{P_1, P_2,\ldots, P_m\}$,
of total complexity~$n$,
the farthest-polygon Voronoi diagram ${\cal FV}({\cal P})$ partitions the plane
into Voronoi regions such that all points in a Voronoi region share the same farthest polygon in ${\cal P}$.
Let $|P|$ be the number of vertices of a polygon $P\in {\cal P}$
and let $|\P|$ be $\sum_{P\in {\cal P}} |P| = n$.

Their algorithm computes the medial-axis ${\cal M}(P)$ for each polygon $P\in {\cal P}$
and refines ${\cal FV}(P,{\cal P})$ by ${\cal M}(P)$.
${\cal M}(P)$ partitions the plane into regions such that
all points in a region share the same closest element of $P$,
where an element is a vertex or an edge of $P$.
In other words,
for each point $v \in \R^2$,
${\cal M}(P)$ provides a shortest path between $v$ and $P$.
Therefore,
the medial axes for ${\cal FV}(P,{\cal P})$,
with $P \in {\cal P}$,
have the same function as the shortest path maps ${\cal SPM}_p$,
with $p \in S$ in the city metric.
By replacing ${\cal P}$ and ${\cal M}(P)$ with $S$ and ${\cal SPM}_p$ respectively,
the divide-and-conquer algorithm of Cheong et al. \cite{CEGGHLLN-11}
can be modified to compute ${\cal FV}(S)$
with respect to the city metric.

\deleted{The major difference between the farthest-polygon Voronoi diagram and
the farthest-site city Voronoi diagram is the structural complexity.
First, the complexity of ${\cal FV}({\cal P})$ is $\Theta(n)$,
while that of ${\cal FV}(S)$ is $\Theta(nc)$.
Second, if $|P|=m$, $|{\cal M}(P)|=\Theta(m)$,
while $|{\cal SPM}_p|=\Theta(c)$.
Therefore, $\sum_{P\in {\cal P}}|{\cal M}(P)|=O(n)$,
while $\sum_{P\in {\cal P}}|{\cal SPM}_p|=O(nc)$.
\andreas{this paragraph seems to be out of place...}
\chihhung{Find another feasible place or just remove it.
The motivation is to clarify a potential question,
why ${\cal FV}(\P)$ can be constructed in $O(n\log^3 n)$ time
but ${\cal FV}(S)$ requires $O(nc\log n \log^2(n+c))$ time.
The answer is their different sizes.}}

Cheong et al. \cite{CEGGHLLN-11} pointed out
the bottleneck with respect to running time is to find for each closed merge curve
a point that lies on it.
In order to overcome the bottleneck,
the authors use some specific point location data structures
\cite{EGS-86,Mulmuley-90}.
Let ${\cal P}$ be divided into two sets ${\cal P}_1$ and  ${\cal P}_2={\cal P}\setminus {\cal P}_1$,
where $|{\cal P}_1|\approx |{\cal P}_2|\approx \frac{n}{2}$.
Cheong et al. \cite{CEGGHLLN-11} construct the point location data structures for ${\cal FV}({\cal P}_1)$ and ${\cal FV}({\cal P}_2)$.
For each polygon $P \in {\cal P}_1$ and each vertex $v \in {\cal M}(P)\cap{\cal FV}(P,{\cal P}_1)$,
they perform a point location query in ${\cal FV}({\cal P}_1)$  and ${\cal FV}({\cal P}_2)$ (likewise for each polygon $P' \in {\cal P}_2$).
Each point location query requires $O(\log n)$ primitive operations,
and each operation tests for $O(1)$ points and takes $O(\log n)$ time.
Hence, one point location query takes $O(\log^2 n)$ time.
Since $|{\cal FV}({\cal P}_1)|=|{\cal FV}({\cal P}_1)|=O(n)$,
merging ${\cal FV}({\cal P}_1)$ and ${\cal FV}({\cal P}_2)$
takes $O(n\log^2 n)$ time.

Since in our case
$|{\cal FV}(S_1)|=|{\cal FV}(S_2)|=O(nc)$,
we perform $O(nc)$ point location queries,
each of which takes $O(\log^2{nc})= O(\log^2(n+c)^2)=O(\log^2{(n+c)})$ time.
Therefore,
merging ${\cal FV}(S_1)$ and ${\cal FV}(S_2)$
takes $O(nc\log^2 (n+c))$ time.
We conclude:

\newcommand{\thmfarthesttimetext}{$\FV(S)$ can be computed in $O(nc\log n\log^2(n+c))$ time.
}

\begin{Theorem}~\label{thm-farthest-time}
\thmfarthesttimetext
\end{Theorem}

\begin{proof}
In the beginning, for each site $p\in S$,
${\cal FV}(\{p\})$ is exactly $\SPM_p$.
Computing $\SPM_p$ takes $O(c\log c)$ time \cite{BKC-09},
implying that computing ${\cal FV}(\{p\})$, for all $p \in S$, takes $O(nc\log c)$ time.
Consider the merge process at some level $i$.
The set $S$ is divided into $2^i$ subsets,
and each of them contains at most %$\frac{n}{2^i}$ sites.
$n/2^i$ sites.
Therefore, the merging process at level $i$
takes $2^i\cdot O(n/2^i\log^2 (n/2^i+c))=O(nc\log^2 (n+c))$ time.
Since there are $\log n$ levels,
${\cal FV}(S)$ can be computed in $O(nc\log n\log^2(n+c))$ time.
\qed
\end{proof}

\section{Conclusion}
\label{sec-conclusion}

We contribute two major results for the \kthorder city Voronoi diagram.
First, we prove that its structural complexity is $O(k(n-k)+kc)$ and $\Omega(n+kc)$.
This is quite different from the $O(k(n-k))$ bound in the Euclidean metric~\cite{Lee-82}.
It is especially noteworthy that when $k=n-1$, i.e., the farthest-site Voronoi diagram,
its structural complexity in the Euclidean metric is $O(n)$,
while in the city metric it is $\Theta(nc)$.
Secondly, we develop the first algorithms that compute the \kthorder city Voronoi diagram
and the farthest-site Voronoi diagram.
Our algorithms show that traditional techniques can be applied to the city metric.
Furthermore,
since the complexity of the first-order city Voronoi diagram is $O(n+c)$,
one may think that the complexity the transportation network contributes to the complexity of the \kthorder city Voronoi diagram is independent of $k$.
However, our results show that the impact of the transportation network increases with the value of $k$
rather than being constant.

%Our results in the structural complexity lead to two interesting observations.
\deleted{
\chihhung{observation is a little weak. Maybe find a new term more contributive.}.
First, in most cases, e.g. general sites (line segments and polygons),
the complexity of the farthest-site Euclidean Voronoi diagram is still linear.
}
\deleted{
However, recently, Bae and Chwa \cite{BC-09} proved that complexity of farthest-site geodesic Voronoi diagram
is $\Theta(nm)$, where $m$ is the total complexity of obstacles.
%However, our results indicate that
First, the results indicate that
the underlying distance metric significantly influences the complexity of the corresponding Voronoi diagrams.
Second,
since the complexity of the first-order city Voronoi diagram is $O(n+c)$,
one may think that the complexity the transportation network contributes to the complexity of the \kthorder city Voronoi diagram is independent of $k$.
However, our results show that the impact of the transportation network increases with the value of $k$
rather than being constant.
Besides, with a transportation network on the Euclidean plane, i.e., the Euclidean city metric,
the size of the corresponding nearest-site Voronoi diagram is already $\Theta(nc)$.
From our results, we make a conjecture that the complexity of the \kthorder Voronoi diagram in the Euclidean city metric is $O(k(n-k)c)$
\chihhung{We also can remove the conjecture.
I originally want to mention a conjecture and prove in the journal version.}.}

\section*{Acknowledgment}
A. Gemsa received financial support by the \emph{Concept for the Future} of KIT within the
framework of the German Excellence Initiative. D. T. Lee and C.-H. Liu are supported by the National Science Council, Taiwan under grants No. NSC-98-2221-E-001-007-MY3 and No. NSC-99-2911-I-001-506.

\deleted{
\newpage
\appendix
\section*{Appendix}

\rephrase{Lemma}{\ref{lem-mix-diagram}}{\lemmixdiagramtext}

%\textbf{Lemma~\ref{lem-mix-vertices}}
%\emph{The number of mixed vertices of $V_k(S)$ is $O(n+kc)$.}
\rephrase{Lemma}{\ref{lem-mix-vertices}}{\lemmixverticestext}

%\textbf{Theorem~\ref{thm-lower-diagram}}
%\emph{The structural complexity of $V_k(S)$ is $\Omega(n+kc)$.}
\rephrase{Theorem}{\ref{thm-lower-diagram}}{\thmlowerdiagramtext}

\rephrase{Lemma}{\ref{lem-compute-region}}{\lemcomputeregion}

%\textbf{Lemma~\ref{lem-compute-diagram}}
%\emph{$V_{j+1}(S)$ can be computed from $V_j(S)$ in $O((j(n-j)+jc)\log(n+c))$ time.}
\rephrase{Lemma}{\ref{lem-compute-diagram}}{\lemcomputediagramtext}

%\textbf{Theorem~\ref{thm-kth-time}}
%\emph{$V_k(S)$ can be computed in $O(k^2(n+c)\log(n+c))$ time.}
\rephrase{Theorem}{\ref{thm-kth-time}}{\thmkthtimetext}

%\textbf{Theorem~\ref{thm-farthest-time}}
%\emph{$\FV(S)$ can be computed in $O(nc\log n\log^2(n+c))$ time.}
\rephrase{Theorem}{\ref{thm-farthest-time}}{\thmfarthesttimetext}

}

\end{document}